\documentclass[12pt]{amsart}
\usepackage{eurosym}
\usepackage{amsfonts}
\usepackage{geometry}
\usepackage{graphicx}
\usepackage{xcolor}
\usepackage{float}
\usepackage{amsthm}
\usepackage{amsmath}
\usepackage{tikz}
\usepackage[hidelinks]{hyperref}

\graphicspath{ {./images/} }
\usetikzlibrary{shadings}
\theoremstyle{plain}

\newcommand{\be}{\begin{equation}}
\newcommand{\ee}{\end{equation}}

\newtheorem{proposition}{Proposition}
\newtheorem{Thm}{Theorem}

\newtheorem{corollary}{Corollary}
\newtheorem{remark}{Remark}
\newtheorem{example}{Example}

\begin{document}
\title{Innovation and Imitation}
\author[J. Benhabib]{Jess Benhabib}
\address{Department of Economics, New York University, 19 West 4th Street,
New York, NY 10003, USA}
\author[E. Brunet]{Eric Brunet}
\address{Laboratoire de Physique de l'\'Ecole normale
sup\'erieure, ENS, Universit\'e PSL, CNRS, Sorbonne Universit\'e,
Universit\'e de Paris, F-75005 Paris, France}
\author[M. Hager]{Mildred Hager}
\address{Department of Economics, New York University, 19 West 4th Street,
New York, NY 10003, USA}
\date{\today{}}

\begin{abstract}
We study several models of growth driven by innovation and imitation by a
continuum of firms, focusing on the interaction between the two. We first
investigate a model on a technology ladder where innovation and imitation
combine to generate a balanced growth path (BGP) with compact support, and
with productivity distributions for firms that are truncated power-laws. We
start with a simple model where firms can adopt technologies of other firms
with higher productivities according to exogenous probabilities. We then
study the case where the adoption probabilities depend on the probability
distribution of productivities at each time. We finally consider models with
a finite number of firms, which by construction have firm productivity
distributions with bounded support. Stochastic imitation and innovation can
make the distance of the productivity frontier to the lowest productivity
level fluctuate, and this distance can occasionally become large.
Alternatively, if we fix the length of the support of the productivity
distribution because firms too far from the frontier cannot survive, the
number of firms can fluctuate randomly.
\end{abstract}

\maketitle

\section{Introduction}

Economic growth is partly the result of costly research activities that
firms undertake in order to innovate, and to increase their productivity.
Growth is also driven by technology diffusion and imitation that takes place
between firms. The role of technology diffusion across countries is
evidenced by the extraordinary sustained growth rates in China and other
East Asian countries during the recent decades. In this paper we investigate
several models of growth driven both by innovation and imitation, focusing
on the interaction between the two. New ideas and innovations push out the
technology frontier. Imitation enables firms to catch up with those higher
up on the technology ladder. We study the dynamics of the productivity
distribution of firms, where productivity is increasing with the rates of
innovation and imitation, and we provide a characterization of its
stationary distribution in the long run. In our study we do not take into
account the effect of the size of the firm on its growth\footnote{One
of the first precursor papers that explores the dynamics of firm-size
distributions is Bonini and Simon (1958). They introduce random growth
proportionate to firm size, coupled with entry of new firms of the
smallest-size at a constant rate. In the limit the productivity distribution
converges to a Pareto distribution. Another classical investigation of the
firm productivity and size distribution is Hopenhayn (1992).}.

As demonstrated by Lucas (2009), in models of technology diffusion based on
imitation alone, growth can be sustained only if the initial distribution of
productivities has unbounded support for high productivity levels. In Lucas
and Moll (2014), and Perla and Tonetti (2014), technology diffusion is
search theoretic, where firms seek higher productivity firms to imitate from
and to adopt superior technology. In these models an unbounded productivity
distribution is necessary to sustain growth through imitation in the long
run. With an initial productivity distribution that has bounded support
imitation ultimately stops, as productivities of the imitating firms
collapse toward the productivity frontier. Therefore, unboundedness is more
than a convenient and inconsequential simplification.

In contrast, in models of endogenous growth, innovation is the primary
driving force of growth. Firms engage in research to generate individual
innovations. These innovations later may become a common stock of ideas that
are available to the whole economy, generating spillovers (Romer (1990)).
Alternatively, innovations are Schumpeterian, in the sense that firms can
leapfrog beyond the productivity frontier. They overtake incumbent firms and
drive them out of business, increasing overall productivity over time
(Aghion and Howitt (1992)).

Models involving both random innovations via a geometric Brownian motion, as
well as imitation via random meetings between firms generating technology
diffusion, have been proposed by Luttmer (2012) and Staley (2011)\footnote{Other
recent models combining innovation and imitation include Benhabib,
Perla and Tonnetti (2014), K\"onig, Lorenz and Zillibotti (2016), Akcigit
and Kerr (2016) and Buera and Lucas (2018).}. Their approach is related to
the KPP\ equation, originally studied in the mathematics literature by
Kolmogorov, Petrovski and Piskunov (1937), and later by McKean (1975) and
Bramson (1984) among others. These models can admit a unique balanced growth
path (BGP) that is a global attractor, and whose shape depends on imitation
and innovation propensities, but not on initial conditions\footnote{To
be more precise however, for the KPP equation the asymptotic BGP velocity
and shape does depend on initial conditions if the initial distribution is
thick tailed. See Bramson (1984).}. Innovations driven by Brownian motion
however assures that the productivity distribution immediately becomes
unbounded, and the resulting BGP does not have compact support.

Having a compact support is particularly relevant for empirical purposes, as
the support of the productivity distribution in individual industries is
found to be quite localized (Syverson (2004), Hsieh and Klenow (2009)).
Firms with significantly low productivity relative to the frontier firms are
unlikely to survive the competition, and to preserve their market shares for
long. The forces of Schumpeterian ``creative destruction'' may endogenously
replace the inefficient firms at the bottom of the productivity
distribution. However other firms, below but not too far from the frontier
may survive, giving a distribution of productivities that allows both for
innovation and imitation to persist over time.

In section~\ref{fixedsupport}, we first investigate a model on a ladder.
Innovation and imitation combine to generate a balanced growth path (BGP)
with compact support. In contrast to models with imitation alone (see Lucas
(2009)), the distribution does not collapse to the frontier either. The
distribution of productivities is centered around some productivity moving
up at a constant growth rate, and keeps its shape relative to this
productivity over time (a traveling wave which is compactly supported). In
section~\ref{sub:exo} we first propose a very simple model of firms on a
quality ladder that can both innovate and imitate, and where with some
positive probability imitators can leapfrog to the productivity frontier. We
characterize the stationary distribution of productivities as a truncated
power law. This model has the advantage of being very simple, but leaves
imitation rates mostly exogenous. In section \ref{density} we extend this
model to introduce density dependent imitation rates. In section~\ref{mlength}
we endogenize the length of the support of balance growth path as
arising from optimal choices of firms.

Another approach to generating productivity distributions that have finite
support is to limit the number of firms to be finite. By construction,
distributions over a finite number of firms have bounded support; however,
stochastic imitation and innovation can make the distance of the
productivity frontier to the lowest productivity level fluctuate, and this
distance can occasionally become quite large. In section~\ref{sec:BRW}, we
study such models with innovation, imitation and a finite number of firms,
the so-called $N$-BRW and $L$-BRW models. These models introduce alternative
approaches to modeling entry, exit, and competition, but also feature
balanced growth paths with compact support. We characterize some features of
their productivity distributions and relate them to results obtained in
earlier sections. Section~\ref{sec:conclusion} concludes.

\section{Innovation and imitation with fixed compact support}

\label{fixedsupport}

In this section, we consider a discrete time model of innovation and
imitation. Innovation can be gradual, by moving up a quality ladder as in
Klette and Kortum (2004) and K\"onig, Lorenz and Zillibotti (2016). But it
can also be a breakthrough, where agents or firms move up from the bottom
and overtake the top, that is they ``leapfrog''. On top of that, agents
imitate other agents. So we start in section~\ref{sub:exo} with a model of
exogenous imitation rates. In section~\ref{density}, we endogenize the
imitation choice and obtain a stationnary productivity distribution that
looks like a truncated power law on a finite support. Then in section~\ref{mlength}
we show that the assumption of a fixed support length of the BGP
is actually the result of an optimal choice problem that trades off the
costs of imitation and its benefits, as in Perla and Tonetti (2014).

In this section, we only consider the case where the number of firms is
sufficiently large to neglect finite-size effects and stochastic behavior.
In fact, to make ``microscopic'' and probabilistic interpretations at the
firm level, and not only speak about densities, we have to assume that a law
of large number holds.

\subsection{Exogenous innovation and imitation}

\label{sub:exo}

At each time $t \in \mathbb{N}$, a firm has a productivity level $i \in 
\mathbb{N}$ on a discrete ladder\footnote{There is no assumption that these productivity levels placed on a ladder are
equally spaced: the rungs of the ladder need not be equidistant from each
other. They simply represent the productivities that can be imitated and
adopted, and we could easily use any ladder for $i$ that maps to $\mathbb{N}$.}.

The density of firms at time $t$ on level $i$ is given by a non-negative
number $f_{t}^{i} \in \mathbb{R}_{+}^{0}$ with $\sum_i f_{t}^{i}=1$ for
every $t$. At each time step, when going from $t$ to $t+1$, firms improve
their productivity and the density climbs up the ladder along some rules
which we explicit now. We assume that, at each time step and at each level,
a fraction $a\in(0,1)$ of firms moves up the productivity ladder by one
level (``innovation''), and a fraction $1-a$ remains stagnant, with the same
productivity. This amounts to assuming a law of large numbers for random
innovation with probability of success $a$. Then, all of the firms that
remained stagnant at the lowest productivity level $i=1$ either leapfrog or
imitate as described below, leaving the lowest level empty. (This
corresponds to a fraction $(1-a)f_t^1$ of all firms.)

We call $m \in \mathbb{N}$ the highest level at time $t$. Given our process,
at time $t+1$, the productivity level $i=m+1$ gets populated, and the lowest
productivity level $i=1$ is emptied as described below. Then, at each
period, we rename the levels: what was the level $i$ at time $t$ becomes the
level $i-1$ at time $t+1$. In this way, the populated levels at the
beginning of each time step are always numbered $\{1,2,\ldots,m\}$. For the
moment, we take the length of support $m $ as given, and postpone a
discussion of endogenously chosen $m$ to section~\ref{mlength}.

In this section~\ref{sub:exo}, imitation for non-innovating firms at the
lowest level happens as follows: at level $i\in\{1,2,\ldots,m\}$, the
fraction of imitators entering level $i$ at time $t+1$ is $(1-a)f_t^1 q_i$,
where the $q_i \in [0,1]$ satisfy 
\begin{equation}
q_i\ge0,\qquad \sum_{i=1}^m q_i=1.
\end{equation}
The $q_i$ for $i\in\{1,\ldots,m-1\}$ represent imitation (exogenous in this
section~\ref{sub:exo}), while $q_m$ represents leapfrogging --- \textit{i.e.}
firms at the lowest level of the productivity distribution at each time $t$
that overtake the current productivity frontier (indeed, the productivity
level $m$ at time $t+1$ corresponds to the productivity level $m+1$ at time $t$
which was not yet populated). If $q_m=0$, leapfrogging is excluded, while
setting $q_m=1$ excludes any imitation.

Note that we assume that jumps to higher productivity levels, whether from
leapfrogging or imitation, do not depend on target productivity densities,
so for the time being we abstract away from any search-theoretic
microfoundation.

The transition dynamics can be written in a single equation 
\begin{equation}  \label{dyn1}
f_{t+1}^i = \underbrace{(1-a)f_t^{i+1}}_\text{fall back} + \underbrace{a
f_t^i }_\text{innovation} + \underbrace{(1-a) f_t^1 q_i}_\text{imitation or
leapfrogging}
\end{equation}
or, more conveniently, with a matrix $A \in M_{m}([0,1])$ representing the
productivity dynamics: 
\begin{align}
f_{t+1} & =A f_{t} \ ,  \notag
\end{align}
\begin{align}
\left[ 
\begin{array}{c}
f_{t+1}^{1} \\ 
f_{t+1}^{2} \\ 
. \\ 
. \\ 
. \\ 
f_{t+1}^{m}\strut
\end{array}
\right] & =\left[ 
\begin{array}{cccccc}
a + q_{1} (1-a) & 1-a & 0 & . & . & 0 \\ 
q_{2}\left( 1-a\right) & a & 1-a & 0 & . & . \\ 
. & . & . & . & . & . \\ 
. & . & . & a & 1-a & 0 \\ 
q_{m-1}\left( 1-a\right) & . & . & 0 & a & 1-a \\ 
q_{m}\left( 1-a\right) & 0 & . & . & 0 & a
\end{array}
\right] \left[ 
\begin{array}{c}
f_{t}^{1} \\ 
f_{t}^{2} \\ 
. \\ 
. \\ 
. \\ 
f_{t}^{m}
\end{array}
\right] \ .  \label{dyn}
\end{align}

By construction $A$ has column sums adding to $1$ as the number of firms
remains constant (in effect, we have a particular birth and death model). $A$
admits $1$ as an eigenvalue. The associated eigenvector is the stationary
distribution for productivity densities, moving up as a traveling wave. The
stationary distribution can be characterized as follows.

\begin{proposition}
\label{exogenous} Let $Q_{s}=\left(
q_{m}+q_{m-1}+\cdots+q_{s}\right)=\sum_{j=s}^{m}q_{j}$, with $Q_m=q_m$ and $Q_1=1$.

The stationary distribution $\left( f_{\infty }^{1},f_{\infty }^{2},\ldots
,f_{\infty }^{m}\right) $, for any $a\in (0,1)$, is given by: 
\begin{equation}  \label{stat}
f_{\infty }^{s}=Q_{s}f_{\infty }^{1}\ ,s=1,...,m\ ,
\end{equation}
and 
\begin{equation*}
f_{\infty }^{1}\left( \sum\limits_{s=1}^{m}Q_{s}\right) =1\ or\ \ f_{\infty
}^{1}=\left( \sum\limits_{s=1}^{m}Q_{s}\right) ^{-1}\ .
\end{equation*}
\end{proposition}

\begin{proof}
The stationary solution fulfills $f_{\infty} = A f_{\infty}$. To simplify
notation let $f_{\infty }^{j}\equiv x_{j},$ $j=1,\ldots,m$.

We start with the last line of equation (\ref{dyn}): 
\begin{equation*}
q_{m}\left( 1-a\right) x_{1}+ax_{m}=x_{m}\qquad \Rightarrow\qquad
x_{m}=q_{m}x_{1}\ .
\end{equation*}

We prove by induction. The line next to last yields 
\begin{align*}
& q_{m-1}\left( 1-a\right) x_{1}+ax_{m-1}+(1-a)x_{m}=x_{m-1} \\
&\qquad \Rightarrow \left( q_{m-1}+q_{m}\right) \left( 1-a\right)
x_{1}=\left( 1-a\right) x_{m-1} \\
&\qquad \Rightarrow x_{m-1}=\left( q_{m-1}+q_{m}\right) x_{1} \ .
\end{align*}
Assume that $x_{m-(s-1)} = (q_{m-(s-1)} +\cdots+q_{m})x_{1}$ . Then we have 
\begin{align*}
& q_{m-s}\left( 1-a\right) x_{1}+ax_{m-s}+\left( 1-a\right)
x_{m-(s-1)}=x_{m-s} \\
& \qquad\Rightarrow x_{m-s}=\left( q_{m-s}+q_{m- (s-1)
}+\cdots+q_{m-1}+q_{m}\right) x_{1} \ .
\end{align*}
This completes the induction proof. Relabeling $m-s$ as $s$, we obtain (\ref{stat}).

We have left one free variable, $x_{1}$, which will be determined by the
normalization of $f$; writing $\sum_{i=1}^{m} x_i =1=x_1 \left(
\sum_{s=0}^{m-1} Q_{m-s}\right) = x_{1}\left( \sum_{s=1}^{m}Q_{s}\right) $,
we get the results.
\end{proof}

The stationary distribution is independent of the probability of innovation $a\in (0,1),$ and only depends on the intensity $q_{i}$ of imitation rates
across productivities. But the speed of convergence to the stationary
distribution depends on $a$, as it affects the eigenvalues of $A$. In
particular the second highest eigenvalue of $A$, which is less than $1$ in
modulus\footnote{Indeed, the matrix $A$ has non-negative entries, is aperiodic since the
diagonal elements are positive, and irreducible as any productivity level
can be reached from any other one. Therefore, the Perron-Frobenius Theorem
implies that the largest eigenvalue --- here, $1$ --- is simple, and that
all other eigenvalues are strictly smaller in modulus.}, can be taken as an
indicator of the convergence rate. The lower this eigenvalue, the faster the
convergence rate. For $m=2$, it can be explicitly computed to be equal to $2a-1 +q_{1}(1-a)=a-(1-a)q_{2}$. This is increasing in $q_{1}$ (more
imitation implies slower convergence), decreasing in $q_{2}$, the
leapfrogging rate (more leapfrogging implies faster convergence), and
increasing in~$a$ (more innovation implies slower convergence).

Firms at any productivity level except the lowest one tend to drop down the
ladder over time. At the bottom of the ladder, non-innovating firms jump to
higher levels through innovation and imitation. Overall, the stationary
density of productivity levels is non-increasing over productivity levels.

We now discuss two special cases:

\subsubsection*{No imitation, only leapfrogging\label{nolf}}

If there is no imitation and only leapfrogging, that is if $q_{m}=1$ and
therefore $q_{i}=0$ for $i\in \{1,\ldots ,m-1\}$, it follows that $f_{\infty
}^{i}=m^{-1}$ for $i=1,\ldots,m$, so the productivity distribution becomes
uniform. Firms that jump to the frontier slide down the productivity
distribution, until they reach the lowest density from which they again jump
to the frontier.

\subsubsection*{No leapfrogging, only imitation}

In this case $q_{m}=0$ and the matrix in (\ref{dyn}) is decomposable. In
particular, the highest productivity evolves with $f_{t+1}^{m} =a f_{t}^{m}$
independently, and converges to zero. This makes the last element of the
eigenvector associated with root $1$ equal to zero, so there is no density
for it at the stationary distribution: $f_{\infty }^{m}=0$.

\subsection{Density dependent imitation}

\label{density}

In Proposition 1 we solved for the densities in terms of exogenous imitation
rates $q_{i},$ $i=1,\ldots,m-1,$ with $m>2$. Now we consider the case that
the imitation rates are proportional to densities. We are again seeking a
stationary solution.

If imitation is similar to learning from another firm, then imitation rates
should be proportional to the number of firms to learn from, or the density
at the corresponding ladder point. Learning is then conditional on meeting
another firm with higher productivity, which happens with a probability
proportional to the density there. Therefore, we let imitation rates be 
\begin{equation}
q_{j}\equiv q_{j}^{t}=\mu f_{t}^{j+1}\ ,\ j=1,\ldots,m-1\ .  \label{dens}
\end{equation}
(Recall that $q_j$ is the probability of jumping to site $j$ at time $t+1$,
which is the same as site $j+1$ at time~$t$ because of the relabeling at
each time step; this is why $q_j$ is proportional to $f_t^{j+1}$ and not $f_t^j$.) Here, $\mu$, which is determined by normalization, is
time-dependent: as seen below, it can be written as a function of $f_t^1$.
The highest $f_{t+1}^{m}$ is not \textit{imitated} because it is not
available for for imitation yet, so $q_{m}$, which represents leapfrogging,
is independent of the densities, as in section~\ref{sub:exo}. Observe that
the problem is now non-linear, so existence and uniqueness of a solution are
more involved than in the linear case. A stationary solution is again $f_{\infty }=A(f_{\infty })f_{\infty }$.

We first determine $\mu$: with $m$ fixed, we must have 
\begin{equation}
1=f_{t}^{1}+\cdots+f_{t}^{m}  \label{norm1}
\end{equation}
and 
\begin{equation}
1=q_{1}+q_{2}+\cdots +q_{m}\ .  \label{norm2}
\end{equation}
Therefore, in order to find a solution, $\mu $ cannot take arbitrary values
but will be determined (together with $f_{\infty }$) as a function of $q_{m}$. Indeed, inserting (\ref{dens}) into (\ref{norm2}), and using (\ref{norm1})
we obtain: 
\begin{equation}
1=q_{m}+\mu \sum_{j=1}^{m-1}f_{t}^{j+1}=q_{m}+\mu (1-f_{t}^{1})\ .
\end{equation}
This implies that, assuming that $f_{t}^{1}<1$ 
\begin{equation}  \label{mut}
\mu \equiv \mu _{t}=\frac{1-q_{m}}{1-f_{t}^{1}}\ .
\end{equation}
Overall, in this subsection, the two parameters $m$ and $q_{m}$ determine
all other quantities, including $\mu$. Note that the reason for which we are
not free to choose $\mu$ is that we insist, in our model, that the lowest
occupied site be emptied at each time step. This condition leads to \eqref{norm2} and then to \eqref{mut}. In section~\ref{sec:BRW}, we briefly
discuss a model where $\mu$ is an arbitrary parameter and the lowest
occupied site is not necessarily emptied at each time step.

For the stationary solution, we write $x_{j} = f_{\infty}^{j}$ as before and 
\begin{equation}  \label{mu}
\mu = \frac{1-q_{m}}{1-x_{1}} \ .
\end{equation}

\begin{remark}
\label{x11} $x_{1}=1$ is never a solution. If $x_{1}=1$, $x_{j}=0$, $\forall
j=2,\ldots,m$. Then, either $q_{m}=1$ and $q_{j}=0$, $\forall j=1,\ldots,m-1$,
in which case the last line of (\ref{dyn}) reads $x_{m}=(1-a)q_{m}x_{1}=0$. For
$a<1$, this is only possible if $q_{m}=0$, a contradiction. Or, if $q_{m}<1$, $\mu =\infty $ and the problem is not well-defined.
\end{remark}

\begin{proposition}
\label{endogenous}Under the assumptions above, with $q_{m}\in \left( 0,1
\right)$, 
\begin{equation}  \label{xi}
x_{i} = q_{m}x_{1}\left( 1+ \frac{1-q_{m}}{1-x_{1}} x_{1}\right) ^{m-i}\ ,
\qquad i=1,\ldots,m , \ 
\end{equation}
where $x_{1} \in [0,1)$ is the unique solution to 
\begin{equation}  \label{maineq}
1= q_{m} \left( 1 + \frac{1-q_{m}}{1-x_{1}} x_{1} \right) ^{m-1} \ ,
\end{equation}
or 
\begin{equation}  \label{x1}
x_{1} = \frac{ (q_{m})^{ - \frac{1}{m-1}} - 1}{(q_{m})^{ - \frac{1}{m-1}} -
q_{m}} \ .
\end{equation}
\end{proposition}

\begin{proof}
We give a recursive proof. The last line of \eqref{dyn} gives again
$x_{m}=q_{m}x_{1}$. Replacing $q_{i}=\mu x_{i+1}$ for $i=1$ to $m-1$ in
(\ref{dyn}), we again proceed by induction and assume that (\ref{xi}) holds
for $i= m-(j-1)$. Then 
\begin{align*}
x_{m-j}&= x_{m-(j-1)} + q_{m-j}x_{1} = x_{m-(j-1)} (1+ \mu x_{1}) \\
\Rightarrow x_{m-j} &= x_{m} (1+\mu x_{1})^{j} = q_{m}x_{1}\left( 1+\mu
x_{1}\right) ^{j} \ .
\end{align*}
This finishes the induction proof. Relabeling $i=m-j$, we obtain (\ref{xi}).

For existence of a solution, we need that 
\begin{equation}
x_{1} = (1+\mu x_{1})^{m-1} q_{m} x_{1}
\end{equation}
or 
\begin{equation}  \label{cons}
(1+\mu x_{1})^{m-1} q_{m} = 1 \ .
\end{equation}
Inserting the expression for $\mu$, (\ref{mu}), this gives equation (\ref
{maineq}). Let us check that the solution thus obtained is normalized; we have 
\begin{equation*}
\sum_{i=1}^{m} x_{i} = \sum_{j=0}^{m-1}(1+\mu x_{1})^{j} q_{m} x_{1}
= \frac{(1+\mu x_{1})^{m}-1}{(1+\mu x_{1})-1} q_{m}x_{1} = \frac{(1+\mu
x_1)-q_m}{\mu} \ ,
\end{equation*}
where we have also used (\ref{cons}). But according to equation (\ref{mu}),
one has $\mu = 1+ \mu x_{1} - q_{m}$, and so we conclude that 
\begin{equation*}
\sum_{i=1}^{m} x_{i} = 1
\end{equation*}
Therefore, a solution to (\ref{cons}) with $\mu $ given by (\ref{mu}) gives
rise to a normalized $x$, as summed up in equation (\ref{maineq}). Inserting
the expression (\ref{x1}) for $x_{1}$ proves existence of a solution. This
finishes the proof of Proposition 2.
\end{proof}

\begin{corollary}
If $q_{m} = 0$, there is no stationary solution.
\end{corollary}

\begin{proof}
For $q_{m}=0$, we have $x_{m}=0$. Using equation (\ref{dyn}), this implies
that 
\begin{equation}
x_{m-1} = \mu 0 x_{1} + (1-a) 0 + a x_{m-1} \ ,
\end{equation}
which implies that $x_{m-1}=0$ for $a<1$. By recursion, $x_{j} =0$ $\forall
j $, and there is no solution.
\end{proof}

While there is no stationary solution for $q_{m}=0$, the limit of the
dynamics may nevertheless converge to a distribution with $x_{1} \to 1$ and
$x_{i} \to 0$ for $i>1$. Recall that $\{x_i\}=\{1,0,0,\ldots\}$ is not a
stationary state, because it would lead to $\mu=0$ and a ill-defined model.
This case is easily illustrated for $m=2$.

\begin{example}
Dynamics for $m=2$, $q_{2}=0$. We start from a density $f_{0}$ with $f^{1}_{0}=1-f^{2}_{0}$ and $f^{2}_{0}>0$ (else, as already pointed out, we
would have $\mu \rightarrow \infty $). Then the dynamics for $f^2_t$ reduce
to 
\begin{equation*}
f^2_{t+1}= a f^2_t
\end{equation*}
and hence 
\begin{equation*}
f^2_t = f^2_0 a^t \to0\quad\text{as $t\to\infty$} \ .
\end{equation*}
By normalization, 
\begin{equation*}
f^1_t = 1-f^2_t=1-f^2_0 a^t \to1\quad\text{as $t\to\infty$}
\end{equation*}
\end{example}

We can observe that because $q_{m}=0$, the upper level of the density is
falling over time as only a fraction $a$, namely the innovators, remains
there each period. In the general case, all upper levels will successively
experience such a decline in population. Because imitation is proportional
to the number of firms present at the productivity level, fewer and fewer
firms will flow into the higher steps of the ladder, which will be
successively depopulated. In the limit, a single ladder step survives.

We would not think that the problematic asymptotic behavior for $q_{m}=0$ is
a major drawback of this model. Surely there are some highly innovative
firms who leapfrog to the highest operational productivity levels, so that
the case $q_{m}=0$ may be economically less interesting.

\begin{example}
We provide numerical illustrations for the stationary densities for $m=10$
and $q_{m}\in \left\{ 0.1\,;\,0.3\,;\,0.5\,;\,0.99\right\}$. The solutions
for $\mu$ and $x_1$ are 
\begin{equation*}
\begin{cases}
\mu=1.1915,\quad x_1=0.2447 & \text{if }q_m=0.1 \\ 
\mu=0.8431,\quad x_1=0.1698 & \text{if }q_m=0.3 \\ 
\mu=0.5801,\quad x_1=0.1380 & \text{if }q_m=0.5 \\ 
\mu=0.0111,\quad x_1=0.1005 & \text{if }q_m=0.99
\end{cases}
\end{equation*}
These values can be computed from $\mu = (q_m)^{-\frac1{m-1}}-q_m$, which is
obtained from \eqref{mu} and \eqref{x1}.

Figure~\ref{statdens} plots the solution for $\left\{ x_{i}\right\}
_{i=1}^{m}$ for these four cases.

\begin{figure}[ht]
\includegraphics[scale=0.3]{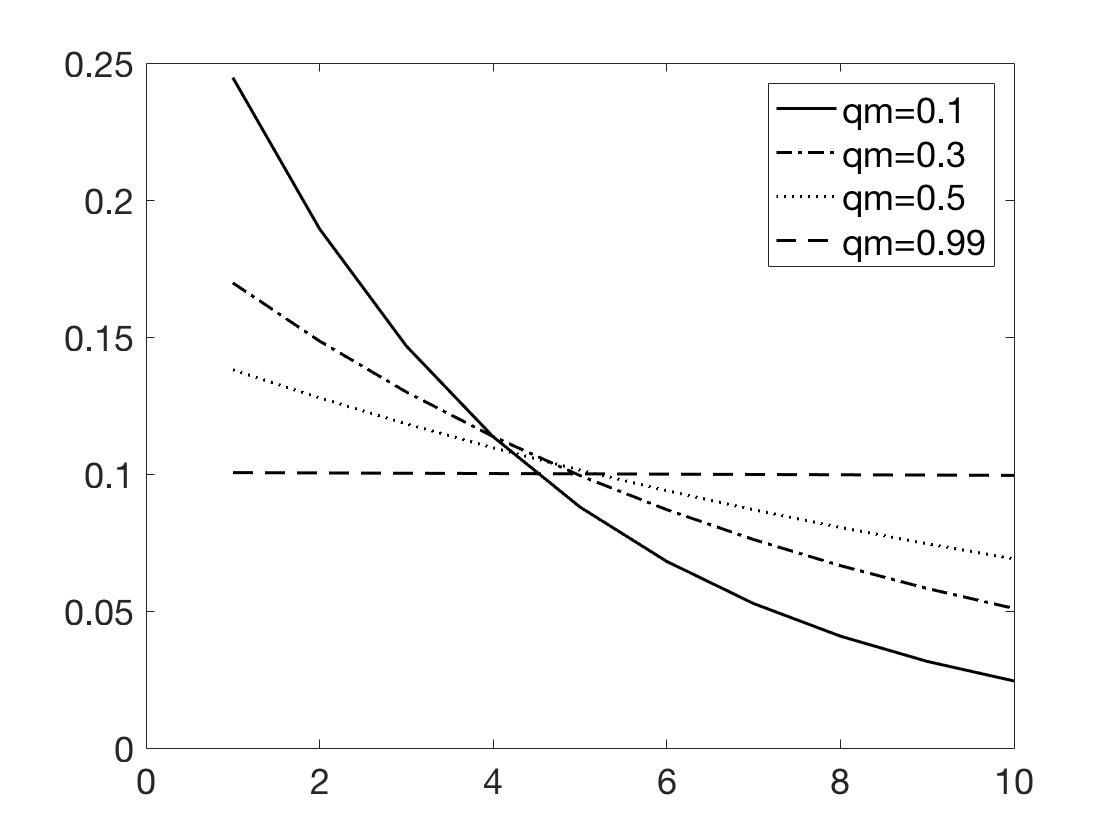} \centering
\caption{Stationary densities for different values of leapfrogging intensity 
$q_{m}$ with $m=10$.}
\label{statdens}
\end{figure}

Notice that higher values of $q_{m}$, or higher leapfrog values, flatten the
productivity distribution. As $q_{m}\rightarrow 1,$ we have $\mu \rightarrow
0$ and\ $x_{i}\to m^{-1}$ for $i=1\,\ldots,m$, so the distribution is
uniform. The stationary density gets increasingly concentrated at the lower
boundary of the productivity ladder as $q_{m}$ gets smaller.

We note that in a continuous time version of this model with a continuum of
firm productivities, with growth driven by imitation as well as by
leap-frogging innovation to the frontier that is governed by a finite Markov
chain, Benhabib, Perla and Tonetti (2017) also show that there exists a
stationary productivity distribution evolving as a travelling wave with
compact and bounded support.
\end{example}

\section{Endogenizing the length of productivity distributions}

\label{mlength}

In the previous section, only the firms at the lowest level $j=1$ would
innovate or imitate at each time step. We now
allow firms at any level $j\leq m$ to choose to
leapfrog or imitate by paying a cost, if the firm estimates that it is
profitable to do so. We first consider the simpler case of leapfrogging
only (no imitation) in section~\ref{leaponly}, and we consider the
full model with leapfrogging and imitation in section~\ref{leapim}.
In the first case (leapfrogging only, section~\ref{leaponly}),
we will show that firms will only
choose to leapfrog or imitate at or below a certain threshold level $j_{0}$
which is independent of time\footnote{As in the previous sections, it is understood that after each time step the
levels are relabeled (so that level $i$ at time $t$ becomes level $i-1$ at
time $t+1$). Then, the highest occupied productivity level is always $m$ at
the beginning of each time step.}. Then with $j_{0}$ determined, the ladder
length is fixed and given by $m-j_0+1$, and the results of section \ref{nolf} with $q_{m}=1$ apply.
In the second case (leapfrogging and imitation, section~\ref{leapim})
we will also show that firms optimally choose
to incur a cost for the opportunity to imitate or leapfrog at or below
a certain threshold level $j_{0}(t)$, but this level might depend on time.
In this case, characterizing the transition dynamics
of the productivity distribution is not straightforward, but if we assume
the system reaches a stationary distribution, then $j_0(t)$ converges to
some $j_0(\infty)$ and the results of Proposition \ref{endogenous} in
section~\ref{density} can be applied with the distribution of firms supported
on an interval of size $m-j_0(\infty)+1$.

In either case, the finite size of the support is now endogenized and
depends on what is the cost to imitate, what are the payoffs at each
quality level, etc.

\subsection{Leapfrogging Only}

\label{leaponly}

We assume that a firm still faces an exogenous probability of innovation~$a$
as in sections~\ref{sub:exo} and~\ref{density} but, when the firm fails to
innovate, it is allowed to make a
choice to leapfrog (and pay some cost) or not. The firm's optimal choice problem is to maximize
its value function, \textit{i.e.\@} the expected discounted value of current and future payoff streams net of
costs\footnote{Under the assumption of linear utility, the benefit of payoffs to the firm
are the payoffs themselves. \textquotedblleft Expected\textquotedblright\
refers to the fact that the firm might have to anticipate the future firms density
in order to project imitation probabilities and thus payoffs.
\textquotedblleft Discounting\textquotedblright\ with a constant
intertemporal discounting factor $\beta _{0}$ as usual reflects the fact
that the firm values the future less than the present. For the reader
unfamiliar with dynamic optimal choice problems, we refer for example to
Lucas and Stokey (1989) or Ljungqvist and Sargent (2018).}. The firm's
choice is to choose for every period whether to pay a cost to leapfrog and
benefit from higher payoffs now and in the future, or not to do so.

When evaluating if it is advantageous to take some action or not, a firm
usually needs to anticipate the future distribution in order to have
expectations for imitation probabilities and outcomes. In the case of
leapfrogging only which we consider now, we will see
that it is actually enough to know the position of the frontier $m=m(t)$,
which remains constant (after relabeling) and equal to its initial position $m(0)$. (Without relabeling, we would have $m(t)=m(0)+t$ due to innovation.)
Therefore, in the case of leapfrogging only, the outcome, when the firm
decides whether to leapfrog or not, depends only on the initial value of $m$; it does not depend on the distribution of firms on the quality ladder, nor on time. (Note that this
will no longer be true when we add imitation in section~\ref{leapim}: with
imitation, it is
necessary to anticipate future distribution of firms to make an optimal
choice.)

As is usual in economics, this optimal choice problem can be reformulated in
a recursive way using a Bellman equation, that we will write down below. We
will show here that the firm's optimal choice is to leapfrog if it lies at a
fixed length below the frontier. This fixed length becomes the new support
size, thus providing a microfoundation to the previously exogenous support
size~$m$.

Every time step, at every level, a firm innovates and moves up one ladder
step with probability $a$. The firms that do not innovate have the choice
either to fall behind, or to catch up with the highest productivity level $m$
(after relabeling, or $m+1$ before) by paying a cost. We assume that it is
not possible to ``imitate'' intermediate
levels.

We assume that the payoffs realized by a given firm increase by some factor
$\lambda>1$ each time the firm takes a step on the quality ladder, and we
introduce the \emph{normalized payoffs} $p_j=\lambda^j$ for a firm being at
level $j\in\{1,\ldots,m\}$. \footnote{Remember that levels are relabeled at
each time step, so level $1$ (for
instance) at different times correspond to different quality levels with
different payoffs. At a given time step, the real payoffs of the different
firms can be obtained by multiplying the \emph{normalized payoffs} $p_j$ by
$\lambda^t$.}

If, as the firm distribution moves up the quality ladder, costs to implement
leap-frogging grow at the same rate $\lambda$ as the payoffs, the firm
problem can be reduced to a stationary problem where the normalized payoffs
$p_j=\lambda^j$ and the \emph{normalized cost} $C$ are independent of time.

Firms, in deciding whether to leapfrog or not, compare the costs to the
expected payoffs. As normalized payoffs increase over the ladder, while
normalized costs do not, firms choose to leapfrog if their distance to the
frontier (the level $m$ of the highest performing firm) is larger than some
threshold, and choose not to leapfrog if their distance to the frontier is
smaller than that threshold.

In other words, there must be a certain threshold level $j_0$ such that a
firm chooses not to leapfrog for productivity levels $j=j_{0}+1,\ldots,m$,
but does leapfrog at levels $j\le j_{0}$. We now provide a formal argument.

Let $V_\text{LF}(j)$ be the value of leapfrogging from some level $j$ and
$V_\text{NLF}(j)$ the value of not leapfrogging from this level.
Then, the value of being at productivity level $j$ is 
\begin{equation}  \label{bel1}
V(j) = \max \big\{ V_\text{LF}(j) , V_\text{NLF}(j)\big \} \ .
\end{equation}
The following equation represents the leapfrogging choice. We have 
\begin{align}  \label{bel2}
V_\text{LF}(j) & = p_{j}+\beta aV ( j) +(1-a) \big[ \beta V( m) - C \big],
\end{align}
where $\beta=\lambda\beta_0$, with $\beta_{0} < 1$ the intertemporal
discount factor; we assume that $\beta <1$. This is the Bellman equation for
the leapfrogging value function, which determines the optimal choice
recursively. The first term on the right-hand side is the payoff received
this period. Then, with probability $a$, the firm innovates and moves up one
step from $j$, which after relabeling becomes $j$, and this continuation
value is discounted with $\beta$. With probability $(1-a)$, the firm does
not innovate but decides to leapfrog; the firm moves above the frontier, at
level $m+1$ (which after relabeling becomes level $m$) and pays the cost $C$.

Similarly, 
\begin{equation}  \label{bel3}
V_\text{NLF}(j) = p_{j} +\beta aV ( j) +\beta (1-a) V( j-1 ) .
\end{equation}

Notice that neither \eqref{bel2} nor \eqref{bel3} depend on the densities $f$
or on time; the value $V(j)$ of being at some level $j$ remains constant in
time.

The firm wants to leapfrog from some level $j$ if leapfrogging is
beneficial, \textit{i.e.} 
\begin{gather}  \label{lf1}
V_\text{LF}(j) > V_\text{NLF}(j) \ ,  \tag{i}
\end{gather}
and does not want to leapfrog if 
\begin{gather}  \label{lf2}
V_\text{LF}(j) < V_\text{NLF}(j) \ .  \tag{ii}
\end{gather}

Observe from \eqref{bel2} and \eqref{bel3} that 
\begin{equation}\label{DeltaV}
V_{\text{LF}}(j)-V_{\text{NLF}}(j)=(1-a)\big[\beta V(m)-C-\beta V(j-1)\big].
\end{equation}

We assume from here on that the value function $V(j)$
increases with the productivity level $j$. This property is not an obvious
consequence of Bellman's equation, but it seems clear that any mathematical
solution where $V(j)$ is not increasing cannot reasonably describe
a real-life situation, because it would mean that some firms should degrade
the quality of their production in order to increase their value.

Then, from~\eqref{DeltaV} and using that $V(j)$ increases in $j$, the quantity $V_{\text{LF}}(j)-V_{\text{NLF}}(j)$
decreases with $j$. Hence, if leapfrogging is
beneficial at $j$, it is even more so at $j-1$, $j-2$, etc. Similarly, if
leapfrogging is not beneficial at $j$, it will be even less so at $j+1$, $j+2$, etc.
 In other
words, there must be a threshold level $j_{0}$ such that 
\begin{equation*}
\text{A site leapfrogs if and only if }j\leq j_{0}.
\end{equation*}

Assume we let this system evolve from an initial condition where the highest
occupied site is $m$. At the end of the first time step, site $m+1$ is
occupied (through innovation and leapfrogging) and all sites up to and
including $j_0$ are emptied through leapfrogging. At the start of the second
time step, after relabeling, the system occupies a subset of sites $\{j_0,\ldots,m\}$. Then, 
at each following time step, only site $j_0$ gets
emptied through leapfrogging and the system remains in $\{j_0,\ldots,m\}$
after relabeling. In the large time limit, the system reaches its stationary
state, which is a uniform distribution over $\{j_0,\ldots,m\}$.

This behavior we have just described is very similar to the behavior of the
system in section~\ref{sub:exo} with $q_m=1$, except that the lowest
occupied site is now $j_0$ instead of 1 in section~\ref{sub:exo}. In other
words, the size of the support is now $m-j_0+1$ instead of $m$. This size of
support depends on the parameters of the model: $a$, $\lambda$, $C$ and $\beta$. (Using 
invariance by translation, it is easy to see that $m-j_0+1$
does not depend on $m$.) This means that the size of the support result from
an endogenized optimum between costs and expected payoffs. By adjusting the
values of the different parameters, any size of support can be obtained.

Through invariance by translation, one can shift the whole system on the
value scale so that the support is on $\{1,\ldots ,m^{\prime }=m-j_{0}+1\}$.
Then, the model is even more similar to section~\ref{sub:exo} with $q_{m}=1$, with the lowest 
occupied level at $j=1$ and with the endogenized $m^{\prime }$ being both the highest 
occupied site and the size of the
support.

\subsection{Leapfrogging and imitation}

\label{leapim}

We now introduce density-dependent imitation as in the section~\ref{density}. A firm innovates at no cost with probability $a$ and, if it does not
innovate, it can choose to pay a fixed cost $C$ to randomly leapfrog or
imitate with density-dependent imitation probabilities, or it can forgo this
opportunity. We are already assuming mean-field dynamics, which are valid in
the limit of a large number of firms. We can therefore safely assume that
the choice taken by a single firm does not impact the distribution.

Recall the following assumptions made in section~\ref{density}: at time $t$,
when a firm chooses to innovate or leapfrog, it jumps with probability $q_{j} $ ( with $j\in \{1,\ldots ,m\}$) onto site $j+1$ which, after
relabeling, becomes site $j$ at time $t+1$. Then, $q_{m}$ is the
(exogenously given) probability of leapfrogging (since site $m+1$ is empty
at time $t$) and $q_{j}$ for $j\leq m-1$ is the probability of imitating. We
assume that for $j\leq m-1$, the probability of jumping on site $j+1$ (site $j$ after relabeling) is proportional to the proportion of firms $f_{t}^{j+1}$
on that site at the beginning of the time step, and we write $q_{j}=q_{j}(f)=(1-q_{m})f_{t}^{j+1}$ for $j\leq m-1$. The value of $1-q_{m}$
of the prefactor is chosen in such a way that the probabilities are
normalized: $\sum_{j\leq m}q_{j}(f)=1$.

The value $V(j;f)$ of being at a site $j$ now depends on the density $\{f_{t}^{k}\}$ of firms at all sites for the current time. As in section~\ref{leaponly}, we write $V_{NLF}(j;f)$ for the value of being at $j$ and
choosing not to imitate/leapfrog given the current densities $f=\left\{
f^{k}\right\} $, and $V_{LF}(j;f)$ for the value of being at site $j$ and to
imitate/leapfrog. The Bellman equations become 
\begin{align}
V(j;f_{t})& =\max \Big(V_{NLF}(j;f_{t}),V_{LF}(j;f_{t})\Big), \\
V_{NLF}(j;f_{t})& =p_{j}+\beta aV(j;f_{t+1})+\beta (1-a)V(j-1;f_{t+1}), \\
V_{LF}(j;f_{t})& =p_{j}+\beta aV(j;f_{t+1})+(1-a)\bigg(\beta
\sum_{k=j}^{m}q_{k}(f_t)V(k;f_{t+1}) \\
& \hspace{5.5cm}+\beta \sum_{k<j}q_{k}(f_t)V(j-1;f_{t+1})-C\bigg)  \notag \\
& =p_{j}+\beta aV(j;f_{t+1})+(1-a)\bigg(\beta V(j-1;f_{t+1})-C\\
& \hspace{2.5cm}+\beta \sum_{k=j}^{m}q_{k}(f_{t})\big[V(k;f_{t+1})-V(j-1;f_{t+1})\big] 
\bigg) 
\notag
\end{align}
(We assume that when it chooses to imitate/leapfrog, the firm first pays the
price $C$. Then it selects a target to imitate with the probabilities $q_{k}$, but if that target is below its own quality, then the firm finally decides
to not move. Recall that $p_{j}=\lambda^j$, $C$, and $\beta$ are normalized
quantities (respectively payoffs, cost and discounting factor) as
in the previous section. The second equality for $V_{LF}(j;f_{t})$ is
obtained using $\sum_{k}q_{k}(f_{t})=1$.)

Notice that 
\begin{equation*}
V_{LF}(j;f_{t})-V_{NLF}(j;f_{t})=(1-a)\bigg(\beta \sum_{k=j}^{m}q_{k}(f)\big[V(k;f_{t+1})-V(j-1;f_{t+1})\big]-C\bigg)
\end{equation*}
and, then, that 
\begin{align*}
& \Big[V_{LF}(j;f_{t})-V_{NLF}(j;f_{t})\Big]-\Big[V_{LF}(j+1;f_{t})-V_{NLF}(j+1;f_{t})\Big] \\
& \qquad \qquad \qquad =(1-a)\beta \bigg(\sum_{k=j}^{m}q_{k}(f_{t})\bigg)
\big[V(j;f_{t+1})-V(j-1;f_{t+1})\big]
\end{align*}
We assume, as in the previous section, that  $V$ increases with $j$.
Then 
$V_{LF}(j;f_{t})-V_{NLF}(j;f_{t})$ decreases with $j$. We conclude, as in
section~\ref{leaponly}, that if imitaton or leapfrogging is advantageous at
a given level $j$, it is even more advantageous at lower levels and that
there must be, for each time $t$, a threshold level $j_{0}(t)$ such that a
firm imitates or leapfrogs if $j\leq j_{0}(t)$ and does not otherwise.

Note that $j_{0}(t)$ depends on time because the values $V(k;f)$ and the
jumping probabilities $q_{k}(f)$ depend on the densities $f$ which depend on
time. In the stationary regime (assuming the system does reach a stationary
regime), $j_{0}(t)=j_{0}$ becomes constant, the firms
live in $\{j_{0},\ldots ,m\}$ and only site $j_{0}$ decides to imitate or
leapfrog. Therefore, the stationary regime is the same as the one described
in section~\ref{density}, but with an endogenized support size, $m-j_{0}+1$.

We have not discussed the case of separate choices for endogenous innovation
and imitation, as analyzed by various authors, like Jovanovic and Rob
(1989), Romer (1990), Aghion and Howitt (1992), Hoppenhyn (1992), Segestrom
(1991), Klette and Kortum (2004), Luttmer (2007), Lucas and Moll (2014),
Konig et al (2016), Benhabib, Perla and Tonetti (2014, 2017). In our ladder
model with firms innovating from every level, separate costs might be chosen
trivially such that firms always want to only innovate or only imitate.

\section{Models with a finite number of firms: \texorpdfstring{$N$}{N}-BRW
and \texorpdfstring{$L$}{L}-BRW}

\label{sec:BRW}

In the models presented in the previous sections, the quantity $f_t^j$
represented the \emph{fraction} of firms with a quality $j$ at time $t$. If
the market is made of $N$ firms, then the number of firms with quality $j$
at time $t$ should be around $Nf_t^j$. If $N\to\infty$, then the number of
firms at a given quality is large for any $j$, the dynamics of the system is
dominated by average quantities and the evolution is deterministic. All the
models presented so far were assumed to be in this $N\to\infty$ limit. In
this section, we explore the effect of having a large but finite number of
firms $N$.

With a finite number of firms, the evolution of the system is intrinsically
stochastic. Consider for instance a single site $j$ containing $n_t^j$ firms
at time $t$, and assume each firm has a probability $a$ of innovating during
one time step. Then, the number of firms innovating is a Binomial random
number of parameters $n_t^j$ and $a$. On average, $an_t^j$ firms innovate
with a standard deviation $[a(1-a)n_t^j]^{1/2}$. On productivity levels
where $an_t^j\gg1$, the fluctuations are negligible compared to the average
behavior and stochasticity can be ignored. On the other hand, when $an_t^j$
is of order~1, the number of innovating firms is essentially random.

The models we consider in this section are stochastic versions of the model
described in section~\ref{density}. We still assume that firms live on a
discrete quality ladder and that time is discrete. At the beginning of any
time step, an active firm is characterized by its productivity level. Then,
during one time step, for each firm, two things can happen (independently).

\begin{itemize}
\item The firm can innovate with probability $a$, thus gaining one
productivity level.

\item The firm can be imitated with probability $\mu$ by a new entrant.
\end{itemize}

The four outcomes for a single firm are graphically represented in figure~\ref{outcomes}. 
\begin{figure}[ht]
\centering
\begin{tikzpicture}
\draw[<->] (0,5.5) node[above, xshift=1.4cm]{productivity level}
--(0,-.2)--(2.8,-.2)
node[below,xshift=-.3cm]{time};
\begin{scope}[xshift=7mm,yshift=5cm]
\node(a)[draw, circle, shading=ball,outer sep=2pt] {};
\node(b)[draw, circle, shading=ball,outer sep=2pt] at (1.5,0){};
\draw[->] (a)--(b);
\node[right] at (2,0) {no innovation, not imitated.};
\node[right] at (7.5,0) {Probability $(1-a)(1-\mu)$};
\end{scope}
\begin{scope}[xshift=7mm,yshift=4cm]
\node(a)[draw, circle, shading=ball,outer sep=2pt] {};
\node(b)[draw, circle, shading=ball,outer sep=2pt] at (1.5,-.1){};
\node(c)[draw, circle, shading=ball,outer sep=2pt] at (1.5,.1){};
\draw[->] (a)--(b);
\draw[->] (a)--(c);
\node[right] at (2,0) {no innovation, imitated.};
\node[right] at (7.5,0) {Probability $(1-a)\mu$};
\end{scope}
\begin{scope}[xshift=7mm,yshift=2.1cm]
\node(a)[draw, circle, shading=ball,outer sep=2pt] {};
\node(b)[draw, circle, shading=ball,outer sep=2pt] at (1.5,1){};
\draw[->] (a)--(b);
\node[right] at (2,.3) {innovation, not imitated.};
\node[right] at (7.5,.3) {Probability $a(1-\mu)$};
\end{scope}
\begin{scope}[xshift=7mm,yshift=.3cm]
\node(a)[draw, circle, shading=ball,outer sep=2pt] {};
\node(b)[draw, circle, shading=ball,outer sep=2pt] at (1.5,0){};
\node(c)[draw, circle, shading=ball,outer sep=2pt] at (1.5,1){};
\draw[->] (a)--(b);
\draw[->] (a)--(c);
\node[right] at (2,.3) {innovation and imitated.};
\node[right] at (7.5,.3) {Probability $a\mu$};
\end{scope}
\end{tikzpicture}
\caption{The four outcomes after a time step for a single firm.}
\label{outcomes}
\end{figure}
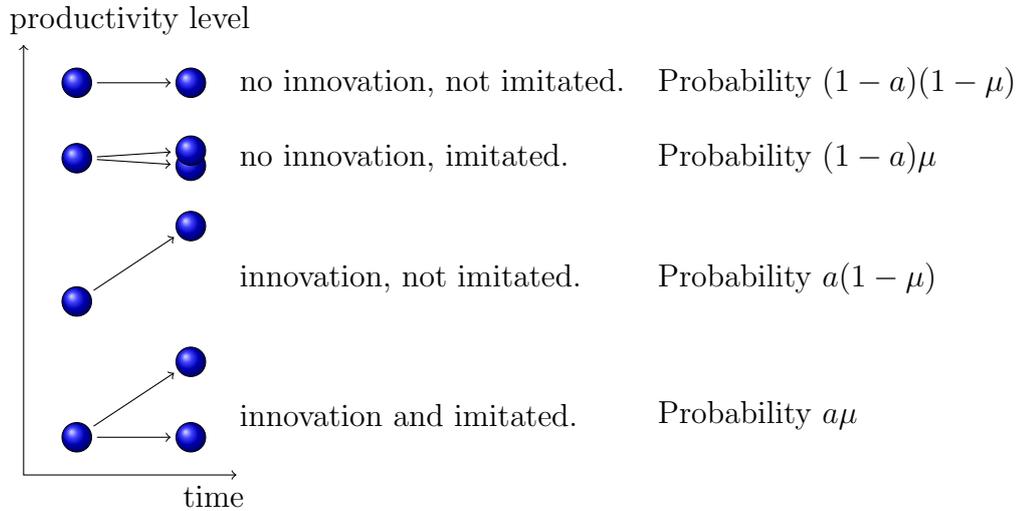

Note that in this section, and unlike in section~\ref{fixedsupport}, we do
not rename the productivity levels after each time step and we assume that $\mu$ is a parameter given exogenously.

The evolution of the whole system during one time step then comes in two
phases: 
\begin{equation}
\begin{cases}
\text{(a)} & \text{each firm present at time $t$ evolves independently
according} \\[-.5ex] 
& \text{to the probabilities in figure~\ref{outcomes},} \\ 
\text{(b)} & \text{a culling of the firms in the system occurs by removing}
\\[-.5ex] 
& \text{some firms at the bottom of the productivity scale.}
\end{cases}
\label{2phases}
\end{equation}

Note that in the evolution phase, the imitating firms can either be some
firms at the lowest productivity level who successfully imitate those above
them, or new entrants displacing firms at the lowest level of productivity.
There is no leapfrogging in this model.

We consider two variants of the model, depending on the way the culling
occurs. A first variant is to fix the number of firms at each time step to
an exogenous parameter $N$. Then, the number of removed firms during the
culling phase must be equal to the number of imitated firms in the evolution
phase to keep the total number of firms constant. This model is called a $N$-BRW ($N$ Branching Random Walk) and is discussed in section~\ref{NBRW}.

Another possibility for the culling phase is to remove all firms lagging $L$
productivity steps or more behind the most productive firm, with $L$ given
exogenously. In this variant, the total number of firms fluctuate with time.
This model is called a $L$-BRW ($L$ Branching Random Walk) and is discussed
in section~\ref{LBRW}.

In the following sections we characterize the shape and properties of the
productivity distributions in the $N$-BRW and $L$-BRW models of innovation
and imitation. The discussion is adapted from works that have been conducted
on KPP fronts since the late nineties in the context of statistical
mechanics, reaction-diffusion models and population genetics. A good point
of entry on this literature is Brunet (2016).

\subsection{The \texorpdfstring{$N$}{N}-BRW model}

\label{NBRW}

Before introducing the $N$-BRW, we need first to discuss what a BRW is. A
Branching Random Walk is a process in discrete time started from a single
particle at the origin. At each time step, each particle (each ``parent'')
is replaced by a random number of particles (the ``children'') positioned
relatively to the parent according to some point process. This rule is
applied independently at each generation for each particle.

\begin{figure}[!ht]
\centering
\begin{tikzpicture}[x=1.5cm]
\draw[<->] (0,3.1) node[above, xshift=.4cm]{position}
--(0,.2)--(3.8,.2)
node[below=-.5pt,xshift=-.3cm]{time};

\node(a)    [draw, circle, shading=ball,outer sep=2pt] at (.5,.5) {};
\node(aa)   [draw, circle, shading=ball,outer sep=2pt] at (1.5,.7) {};
\node(ab)   [draw, circle, shading=ball,outer sep=2pt] at (1.5,1.2) {};
\draw[->] (a) -- (aa);
\draw[->] (a) -- (ab);
\node(aba)  [draw, circle, shading=ball,outer sep=2pt] at (2.5,1.3) {};
\draw[->] (ab) -- (aba);
\node(aaa)  [draw, circle, shading=ball,outer sep=2pt] at (2.5,1.9) {};
\node(aab)  [draw, circle, shading=ball,outer sep=2pt] at (2.5,.9) {};
\draw[->] (aa) -- (aaa);
\draw[->] (aa) -- (aab);
\node(abaa)  [draw, circle, shading=ball,outer sep=2pt] at (3.5,1.3) {};
\node(abab)  [draw, circle, shading=ball,outer sep=2pt] at (3.5,1.8) {};
\node(abac)  [draw, circle, shading=ball,outer sep=2pt] at (3.5,2.3) {};
\draw[->] (aba) -- (abaa);
\draw[->] (aba) -- (abab);
\draw[->] (aba) -- (abac);
\node(aaba) [draw, circle, shading=ball,outer sep=2pt] at (3.5,.9) {};
\node(aabb)  [draw, circle, shading=ball,outer sep=2pt] at (3.5,2.8) {};
\draw[->] (aab) -- (aaba);
\draw[->] (aab) -- (aabb);
\node(aaab)  [draw, circle, shading=ball,outer sep=2pt] at (3.5,3.3) {};
\draw[->] (aaa) -- (aaab);

\begin{scope}[xshift=7cm]
\draw[<->] (0,3.1) node[above, xshift=.4cm]{position}
--(0,.2)--(3.8,.2)
node[below=-.5pt,xshift=-.3cm]{time};

\node(a)    [draw, circle, shading=ball,outer sep=2pt] at (.5,.5) {};
\node(aa)   [draw, circle, shading=ball,outer sep=2pt] at (1.5,.7) {};
\node(ab)   [draw, circle, shading=ball,outer sep=2pt] at (1.5,1.2) {};
\draw[->] (a) -- (aa);
\draw[->] (a) -- (ab);
\node(aba)  [draw, circle, shading=ball,outer sep=2pt] at (2.5,1.3) {};
\draw[->] (ab) -- (aba);
\node(aaa)  [draw, circle, shading=ball,outer sep=2pt] at (2.5,1.9) {};
\node(aab)  [draw, circle, shading=ball,ball color=black!40,outer sep=2pt,opacity=.05] at (2.5,.9) {};
\draw[->] (aa) -- (aaa);
\draw[->] (aa) -- (aab);
\node(abaa)  [draw, circle, shading=ball,ball color=black!40,outer sep=2pt,opacity=.05] at (3.5,1.3) {};
\node(abab)  [draw, circle, shading=ball,ball color=black!40,outer sep=2pt,opacity=.05] at (3.5,1.8) {};
\node(abac)  [draw, circle, shading=ball,outer sep=2pt] at (3.5,2.3) {};
\draw[->] (aba) -- (abaa);
\draw[->] (aba) -- (abab);
\draw[->] (aba) -- (abac);
\node(aaba) [draw, circle, shading=ball,ball color=black!40,outer sep=2pt,opacity=.05] at (3.5,.9) {};
\node(aabb)  [draw, circle, shading=ball,ball color=black!40,outer sep=2pt,opacity=.05] at (3.5,2.8) {};
\draw[->,opacity=.2] (aab) -- (aaba);
\draw[->,opacity=.2] (aab) -- (aabb);
\node(aaab)  [draw, circle, shading=ball,outer sep=2pt] at (3.5,3.3) {};
\draw[->] (aaa) -- (aaab);
\end{scope}
\end{tikzpicture}
\caption{Left: an exemple of BRW where, at each generation a particle can
have 1, 2 or 3 offsprings. Right: a $N$-BRW with $N=2$ obtained by keeping
at each time step only the two highest children of the surviving particles
of the previous time step. Notice that this rule is not the same as keeping
the two highest particles of the BRW at each time step.}
\label{BRW}
\end{figure}
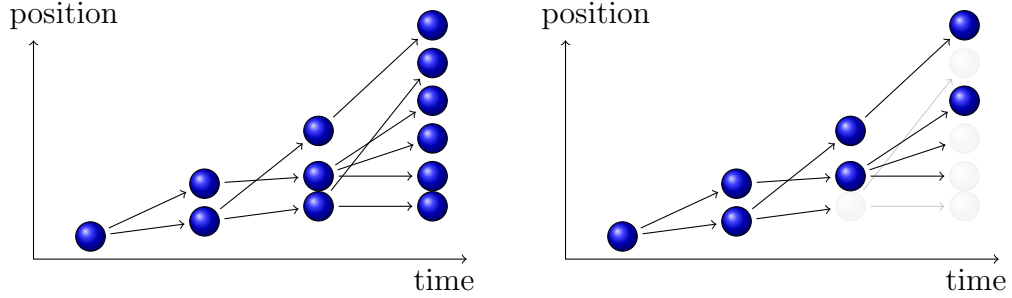

For instance, following figure~\ref{outcomes}, the rule could be that a
particle at $y$ gives either one particle at $y$, or two particles at $y$,
or one particle at $y+1$ or two particles respectively at $y$ and $y+1$. The
left part of figure~\ref{BRW} shows a BRW with a different rule where each
parent can have 1, 2 or 3 children.

Note that the number of particles $N_t$ at each generation follows the
following recursion: $N_{t+1}= \sum_{i=1}^{N_t} n^{i,t}$, where $n^{i,t}$ is
the number of children of individual $i$ at time $t$ and where it is assumed
that the $n^{i,t}$ are independent identically distributed random variables
over integers. This is called a Galton-Watson process. In other words, a BRW
is a Galton-Watson process where we keep as an extra information the
position of the particles. For simplicity, we exclude the possibility that a
particle has zero children and we insist that it has more than one child
with positive probability. Then, the population size increases exponentially
with time.

Denote by $(\epsilon_1, \epsilon_2,\ldots,\epsilon_n)$ the positions of the
children relative to the parent (both $n$ and the $\epsilon_i$ are random).
Then, under conditions on the laws of $n$ and $\epsilon_i$, listed for
example in Gantert, Hu and Shi (2011)\footnote{The conditions are: \newline
\hbox{\quad}a) $\mathbb{E}[n]>1$ (we excluded the case where $n$ can be 0,
and insisted that $n>1$ with positive probability, so this is automatic in
our case.) \newline
\hbox{\quad}b) there exists $\delta>0$ such that $\mathbb{E}
[n^{1+\delta}]<\infty$ (in other words, there are never too many children.
This is automatic if the number of children is bounded.) \newline
\hbox{\quad}c) there exists $\delta>0$ such that $\mathbb{E}\big[
\sum_{i=1}^n e^{\delta\epsilon_i}\big]<\infty$ (in other words, the children
are not created too much upwards relative to the parent. This is automatic
if the number of children $n$ and the displacements $\epsilon_i$ are
bounded, as in our case.) \newline
\hbox{\quad}d) there exists $\delta>0$ such that $\mathbb{E}\big[
\sum_{i=1}^n e^{-\delta\epsilon_i}\big]<\infty$ (in other words, the
children are not created too much downwards relative to the parent. This is
automatic if the number $n$ of children and the displacements $\epsilon_i$
are bounded.) \newline
\hbox{\quad}e) The function $v(\gamma)=\frac1\gamma \log \mathbb{E}\big[
\sum_{i=1}^n e^{\gamma \epsilon_i}\big]$, which is necessarily well defined
on some interval $(0,\delta)$ with $\delta\in(0,\infty]$, must reach a
minimum $v_c=v(\gamma_c)$ on that interval. It is automatic in the example
developed below for any $\mu\in(0,1]$ and $a\in(0,1)$, but for other
problems it might not be automatic.} (see also there for references), one
can show that the highest position $y_\text{max}(t)$ in the BRW at time $t$
increases linearly with time: 
\begin{equation}
\lim_{t\to\infty}\frac{y_\text{max}(t)}{t} = v_c,  \label{def_vc}
\end{equation}
with some velocity $v_c$ given by 
\begin{equation}
v_c = \min_\gamma v(\gamma)=v(\gamma_c) \qquad\text{with }
v(\gamma)=\frac1\gamma \log \mathbb{E}\bigg[\sum_{i=1}^n e^{\gamma
\epsilon_i}\bigg],  \label{vg}
\end{equation}
as soon as this minimum exists for some $\gamma_c>0$.

Here the expectation is both on the displacements $\epsilon_i$ and the
number $n$ of children. $\gamma_c$ is the value of $\gamma$ for which the
minimum is reached.

For instance, with the rules of figure~\ref{outcomes}, one checks that 
\begin{equation*}
v(\gamma)=\frac1\gamma\log \big[1+\mu+a(e^\gamma-1)\big].
\end{equation*}
Indeed, 
\begin{align*}
\mathbb{E}\Big[\sum_{i=1}^n e^{\gamma \epsilon_i}\Big]& = (1-a)(1-\mu)\times
e^0+(1-a)\mu\times(e^0+e^0) \\[-2ex]
&\quad\qquad+a(1-\mu)\times
e^\gamma+a\mu\times(e^0+e^\gamma)=1+\mu+a(e^\gamma-1) \ .
\end{align*}

We can now define the $N$-BRW. The evolution for one time step of a $N$-BRW
goes like a BRW, except that after each step only the $N$ highest particles
are kept, the other being removed, so that after some time there are exactly 
$N$ particles in the system at each time step. Note that this rule is not
the same as keeping the $N$ highest of a BRW at each time step; see the
right part of figure~\ref{BRW}.

The $N$-BRW and related models (the $N$-BBM, the stochastic Fisher equation)
have been studied in mathematics, theoretical physics and biology and
several results are known both from non-rigorous and rigorous arguments.

For the $N$-BRW, a striking result is that one can still define a velocity $v_N$ for the highest particle, as in \eqref{def_vc}. This velocity depends
on $N$, converges to $v_c$ as $N\to\infty$, but the speed of convergence is
unexpectedly slow (this is explained in Brunet and Derrida (1997) with a
rigorous proof provided by Berard and Gou\'er\'e (2010) for the case $\mu =
1 $).

\begin{Thm}[Velocity. Berard and Gou\'er\'e (2010)]
For the $N$-BRW with $\mu=1$, we have: 
\begin{equation}
v_N = v_c -\frac{\pi^2 v^{\prime \prime }(\gamma_c)}{2L_0^2} + o\left(\frac{1
}{L_{0}^{2}}\right)  \label{vN}
\end{equation}
with 
\begin{equation}
L_0=\frac1{\gamma_c}\log N,  \label{LN}
\end{equation}
$v_c$, $v(\gamma)$ and $\gamma_c$ defined as in \eqref{vg} and $o({1}/{
L_{0}^{2}})$ a term that is vanishing faster than $1/L_0^2$ as $N \to \infty$
.
\end{Thm}

(Nota: even though a proof is available only in one case, heuristic
arguments and numerical simulations suggest that \eqref{vN} holds in a large
number of cases.)

\subsubsection*{Size of support}

Based on numerical observations and phenomenological theory for closely
related models, it is believed (see Brunet and Derrida (1997) and Brunet,
Derrida, Mueller and Munier (2006)) that after a long time, the system
reaches a stationary regime as seen from the center of mass of the system.
Here, stationary is to be interpreted in a probabilistic sense: while for
finite $N$, there are still fluctuations, the laws determining the system
become stationary. In this stationary regime, the size of support, which is
the difference between the position $y_\text{max}$ of the highest particle
and the position $y_\text{min}$ of the lowest particle, satisfies 
\begin{equation}
L:=y_\text{max}-y_\text{min} = L_0 + \mathcal{O}(1),  \label{sizeofsupport}
\end{equation}
with $L_0$ as in \eqref{LN} and $\mathcal{O}(1)$ is designating a random
variable whose law becomes independent of $N$ in the large $N$ limit.
(Therefore, it will be smaller and smaller as compared to $L_{0}$ when $N\to
\infty$.).

By construction, a finite number of firms assures a productivity
distribution that has a finite support at any fixed time, but what 
\eqref{sizeofsupport} means is that the firms have at all time comparable
productivity levels, and the scenario where some firms stay put while others
diverge at infinity due to innovations cannot occur. However, because the
process is stochastic, there is a probability of a firm with an extended
streak of successful innovations breaking out for a while, so that the
support of productivity distribution may occasionally get large, but after
some time laggard firms will catch-up via imitation and close the gap.

\subsubsection*{Shape of the front}

Another interesting result concerns the typical density of the cloud of
particles in the stationary regime. To simplify the discussion, we assume
that the underlying BRW is the one described in figure~\ref{outcomes}. Then,
the population lives on the lattice, and we introduce $f(y,t)$ the fraction
of particles (or firms) at position (or quality level) $y$ at time $t$.

After the reproduction phase (but before the culling phase, see 
\eqref{2phases}), the \emph{expected} fraction of firms at position $y$ and
time $t+1$ is $(1-a + \mu) f(y,t) + a\, f(y-1,t) $. Then, one could write
the evolution equation as 
\begin{equation}
f(y,t+1)=(1-a + \mu) f(y,t) + a\, f(y-1,t) + (\text{noise}) \quad\text{if $
y>y_\text{min}(t+1)$} \ ,  \label{noisy}
\end{equation}
where $y_\text{min}(t)$ is the position of the lowest firm at time $t$ (the
values of $y_\text{min}(t+1)$ and of $f\big( y_\text{min}(t+1),t+1\big)$ are
obtained by writing $\sum_y f(y,t+1)=1$). The noise term is some random
number with zero expectation and standard deviation of order $\mathcal{O}(
\sqrt{{f}/{N}})$, depending on the density\footnote{The value for $N[f(y,t+1)-f(y,t)]$ before the culling phase is (the number
of new firms innovating from $y-1$) minus (the number of firms innovating to 
$y+1$) plus (the number of imitators). These three terms are independent
Binomial random variables and so one finds that the exact expression for the
standard deviation of the noise term in \eqref{noisy} is $\sqrt{
[f(y-1,t)a(1-a)+f(y,t)a(1-a)+f(y,t)\mu(1-\mu)]/N}$.}.

As $N\to\infty$, in the so-called hydrodynamic limit, the noise term in 
\eqref{noisy} is expected to disappear. While there are no rigorous result
concerning this hydrodynamic limit for the $N$-BRW, such a result exists for
two closely related models, see Durrett and Remenik (2011) and De Masi,
Ferrari, Presutti and Soprano-Loto (2019). In the first model, time is
continuous, and at rate 1 each particle creates an additional particle at a
random distance $\epsilon$; when this occurs, the lowest particle is removed
to keep the population constant. The second model is the $N$-BBM, which can
be described as follows: time and space are continuous. $N$ particles
perform independent Brownian motions. At rate 1, each particle creates an
offspring at its own position ($\epsilon=0$); when this occurs, the lowest
particle is removed to keep the population constant.

The equation obtained in this large $N$ limit, as given by \eqref{noisy}
without the noise term, is reminiscent of the model described in section~\ref{density}. The only remaining difference is that in section~\ref{density},
the imitation rate was tuned at each time step in such a way that $y_\text{min}(t)$ would increase by exactly one unit at each time step. In 
\eqref{noisy} (with or without the noise term), the imitation rate $\mu$ is
fixed exogenously and, depending on its value, the lower bound $y_\text{min}
(t)$ can increase by several units in a time step or take several time steps
to increase by one unit.

The evolution equation is maybe easier to write on $h(y,t)=\sum_{z\ge y}
f(z,t)$, which represents the fraction of firms with a quality level at
least $y$. One checks that 
\begin{equation}
h(y,t+1)=\min\Big[1,(1-a + \mu) h(y,t) + a\, h(y-1,t) + (\text{noise})\Big]
\label{noisy2}
\end{equation}
where, here again, the noise term disappears in the large $N$ limit. Without
the noise term, \eqref{noisy2} is the discrete-time version of the equation
studied in \cite{Bru} which was shown to display most of the characteristics
of the Fisher-KPP equation. With the noise term, it is very similar to the
equations studied in Brunet and Derrida (1997), Brunet and Derrida (2001)
and Brunet, Derrida, Mueller and Munier (2006) papers, as well as Mueller,
Mytnik and Quastel (2011), which is with continuous time and space.

As suggested by Brunet and Derrida (1997), the velocity \eqref{vN} of the
noisy front (and thus of the $N$-BRW) could be obtained to the $1/L_0^2$
order by replacing the noise term in \eqref{noisy2} by a cutoff of order
$1/N $, meaning that after each time step the value of $h$ is set to 0 at all
the positions $y$ where the evolution equation leads to a result smaller
than $1/N$. Furthermore, the shape of the front at large times is for large
$N$ (and hence large $L_0=(\log N)/\gamma_c$), large $z$ and large $L_0-z$
(so that $z$ is not too close to 0 or to $L_0$) approximately given by 
\begin{equation}
h(y_\text{min}(t)+z,t)\approx A L_0 \sin\frac{\pi z}{L_0}\, e^{-\gamma_c z}
\ .  \label{cutoff}
\end{equation}
Notice then that, to leading order, the density $f(y,t)=h(y,t)-h(y+1,t)$ is
given by the same equation with the prefactor $A$ replaced by 
$A(1-e^{-\gamma_c})$ \footnote{An interesting question, which we postpone to another paper, is whether the
sin prefactor can be observed in real data.}.

The shape of the front \eqref{cutoff} is for the front equation~\eqref{noisy}
with the noise replaced by a cutoff. For the $N$-BRW model itself as
described by equation~\eqref{noisy} with its noise term, Brunet, Derrida,
Munier and Muller (2006) give the following non-rigorous phenomenological
description of the model. This description is supported by numerical
simulation and, to some extent, by rigorous work (Maillard (2016)).

In the $N$-BRW, the shape of the front is most of the time given by the
cutoff shape~\eqref{cutoff} plus some small fluctuations. Occasionally,
typically every $\propto L_0^3$ units of time, a huge fluctuation occurs
where the shape of the front is significantly different from \eqref{cutoff}
for about $\propto L_0^2$ units of time. Such a huge fluctuation comes in
the following way: a single particle moves up further than typical (a single
firm innovates a lot in a short time). That particle branches as usual (the
firm is imitated), but its `imitation offspring', \textit{i.e.\@}
its imitators, the imitators of its imitators, etc., are at first rarely
removed from the system because they typically lie above the others (they
have better quality than the other firms). The end effect of such a
fluctuation is that a positive fraction of all the firms are replaced by the
imitation offspring of the highly successful firm that started the
fluctuation. So, to reformulate, based on numerical computations for similar
models, in the stationary regime, a density of firms like \eqref{cutoff} is
expected, while occasionally (every $\propto L_0^3$ units of time), a single
firm innovates a lot and gets imitated by so many firms that it redefines
the industry (in the sense that the innovation is shared by a positive
fraction of the agents). The transition time to redefine the industry is of
order $\propto L_0^2$.

This is particularly interesting: at random times, a firm is so successful
that a full fraction of the industry ends up imitating it (or its imitators).

A word of caution: the results presented above are asymptotic results, which
are believed to be valid for large values of $N$. It is not obvious that
$N=10^4$ or $N=10^5$ are big enough for these results to be very accurate.

\subsection{The \texorpdfstring{$L$}{L}-BRW model}

\label{LBRW}

A variant of the $N$-BRW is the $L$-BRW which might be more adapted to
describe a situation where lagging firms go out of business and new firms
enter the market. The evolution phase of the $L$-BRW (innovation and
imitation) is the same as for the BRW or the $N$-BRW (particles innovate and
are imitated), but the culling phase is different; in the $L$-BRW, at each
time step, firms with a productivity lagging more than $L$ level below the
leading firm are removed from the system, as it is assumed that they are not
productive enough to survive the market. Here, the parameter $L$ is given
exogenously.

In the $L$-BRW, the number $N$ of firms fluctuates. However, for large $L$,
one observes that the number $N$ of firms fluctuates around some average
value $N_0$ with 
\begin{equation}
N_0= e^{\gamma_c L}  \label{NL}
\end{equation}
which is formally the same relation as \eqref{LN}.

The heuristic argument of Brunet, Derrida, Mueller, and Munier (2006) is
that a $N$-BRW and $L$-BRW have very similar typical behaviors: in the
$N$-BRW, $N$ is a given parameter and $L$ (defined as the observed support size
or distance between the best and worst firm) fluctuates, while in the
$L$-BRW, it is the support size $L$ which is given, and the population size $N$
fluctuates. In either case, one has the relation 
\begin{equation*}
L \approx \frac1{\gamma_c}\log N.
\end{equation*}

Then, one expects that the velocity $v_L$ of the $L$-BRW is given by 
\begin{equation}  \label{vL}
v_L\approx v_c-\frac{\pi^2v^{\prime \prime }(\gamma_c)}{2L^2}
\end{equation}
(compare to \eqref{vN}), that the average shape of the front is given by the
sine shape \eqref{cutoff} of the cutoff theory, etc.

There is unfortunately no rigorous paper establishing these results for the
$L$-BRW. However, Pain (2016) has established result~\eqref{vL} in the case
of the $L$-BBM (where BBM stands for Branching Brownian Motion) which is a
continuous version of the $L$-BRW. More precisely, in the $L$-BBM, particles
perform Brownian motions. With rate~1, a particle is replaced by two
particles, and any particle at a distance more than $L$ from the highest
particle is removed. The fact that \eqref{vL} holds for the $L$-BBM and the
close similarity between the $L$-BBM and the $L$-BRW is a strong indication
that the heuristic arguments are correct.

\section{Conclusion}

\label{sec:conclusion}

We model the dynamics of technology diffusion to characterize shapes of the
stationary firm productivity distributions with a skew, and explore
conditions that will lead to compact productivity supports. Innovation and
imitation activities move the productivity distribution forward, and compact
supports can be sustained as competition causes the low productivity firms
to exit. Section~\ref{fixedsupport} provides a model generating skewed
productivity distributions with compact support. It relies on an
endogenously determined finite productivity ladder, sustained by a fraction
of firms that can leapfrog to the frontier. Section~\ref{sec:BRW} introduces
models with either a finite number of $N$ firms, or a finite productivity
support $L$. In both cases the support of the productivity distribution
remains compact; in the former case the length of the support is stochastic
while the number of firms are constant, and in the latter the support length
is fixed while the number of firms fluctuates.

\end{document}